\documentclass[runningheads]{llncs}
\usepackage[T1]{fontenc}
\usepackage{graphicx}
\usepackage{dirtytalk}
\usepackage{color, soul}
\usepackage{amssymb, acronym,amsfonts,mathtools}
\usepackage{mathrsfs}
\usepackage{enumitem}
\usepackage{bm}
\usepackage{hyperref}  
\definecolor{myblue}{RGB}{0, 70, 180}
\hypersetup{colorlinks,
            linkcolor=myblue,
            citecolor=myblue,
            urlcolor=myblue,
            anchorcolor=myblue,
            filecolor=myblue,
            plainpages=false,
            breaklinks=true}

\newcommand{\bE}{\ensuremath{\mathbb{E}}}

\newcommand{\bR}{\ensuremath{\mathbb{R}}}

\newcommand{\cD}{\ensuremath{\mathcal{D}}}
\newcommand{\cE}{\ensuremath{\mathcal{E}}}

\newcommand{\cP}{\ensuremath{\mathcal{P}}}

\newcommand{\cU}{\ensuremath{\mathcal{U}}}

\newcommand{\cW}{\ensuremath{\mathcal{W}}}
\newcommand{\cX}{\ensuremath{\mathcal{X}}}
\newcommand{\cY}{\ensuremath{\mathcal{Y}}}

\newcommand{\Ren}{R\'enyi }

\definecolor{kth-blue}{RGB}{25,84,166}
\definecolor{kth-pink}{RGB}{216,84,151}
\definecolor{kth-gray}{RGB}{101,101,108}

\DeclarePairedDelimiter\abs{\lvert}{\rvert}%
\DeclarePairedDelimiter\norm{\lVert}{\rVert}%

\allowdisplaybreaks

\begin{document}
\title{Evaluating Differential Privacy on Correlated Datasets Using Pointwise Maximal Leakage}
\titlerunning{Evaluating DP on Correlated Datasets Using PML}
%
\author{Sara~Saeidian \and Tobias~J.~Oechtering \and Mikael~Skoglund}
\authorrunning{S.~Saeidian et al.}
%
\institute{KTH Royal Institute of Technology, 100 44 Stockholm, Sweden\\
\email{\{saeidian,oech,skoglund\}@kth.se}}
\maketitle              
\begin{abstract}
Data-driven advancements significantly contribute to societal progress, yet they also pose substantial risks to privacy. In this landscape, \emph{differential privacy} (DP) has become a cornerstone in privacy preservation efforts. However, the adequacy of DP in scenarios involving correlated datasets has sometimes been questioned and multiple studies have hinted at potential vulnerabilities. In this work, we delve into the nuances of applying DP to correlated datasets by leveraging the concept of \emph{pointwise maximal leakage} (PML) for a quantitative assessment of information leakage. Our investigation reveals that DP's guarantees can be arbitrarily weak for correlated databases when assessed through the lens of PML. More precisely, we prove the existence of a pure DP mechanism with PML levels arbitrarily close to that of a mechanism which releases individual entries from a database without any perturbation. By shedding light on the limitations of DP on correlated datasets, our work aims to foster a deeper understanding of subtle privacy risks and highlight the need for the development of more effective privacy-preserving mechanisms tailored to diverse scenarios.

\keywords{Pointwise maximal leakage  \and Differential privacy \and Correlated data.}
\end{abstract}
\section{Introduction}
\label{sec:intro}
In today's data-driven landscape, private and public organizations increasingly rely on data collected from individuals for decision-making and to enhance service provision. While insights obtained from data undoubtedly offer value to societies, it is essential not to overlook the risks to \emph{privacy} that come along with it. As a result, extensive research over the past decades has led to the emergence of various privacy definitions and frameworks.

Among the various definitions proposed, \emph{differential privacy} (DP)\cite{dworkDifferentialPrivacy,dworkCalibratingNoiseSensitivity} stands out as the most widely accepted framework for understanding and enforcing privacy. DP has been adopted by both public agencies, such as the U.S. Census Bureau~\cite{abowd2018us}, and major industry players like Apple~\cite{thakurta2017learning}, Google~\cite{erlingsson2014rappor}, and Microsoft~\cite{ding2017collecting}. DP assumes that data collected from individuals is stored in a database that returns answers to queries in a privacy-preserving manner. Its objective is to reveal population-level insights about the data while preserving the privacy of each individual. Specifically, DP ensures that two databases differing in a single entry, presumably information pertaining to a single individual, cannot be distinguished based on their corresponding query responses. This approach aligns with the fundamental idea that \say{nothing should be learnable about an individual participating in a database that could not be learned without participation}~\cite{dworkDifficultiesDisclosurePrevention2010a}. 

Despite its widespread success, several studies have raised concerns that DP may not provide sufficient protection for databases containing \emph{correlated} data~\cite{kiferNoFreeLunch2011,cormode2011personal,he2014blowfish,liu2016dependence,yang2015bayesian,liMembershipPrivacyUnifying2013,zhu2014correlated}. Informally, this is because there may be no one-to-one mapping between individuals and entries in the database, and each person's information may contribute to multiple entries. To illustrate this issue, consider the following example from~\cite{kiferNoFreeLunch2011}. 

\begin{example}
\label{ex:bob}
Suppose Bob is part of a medical database where his sensitive attribute can take one of the values $1, \ldots, k$. Assume the database is sampled from a distribution such that when Bob's sensitive attribute is $j$, there are $j \times 10,000$ cancer patients in the data. Suppose an adversary queries the database about the number of cancer patients. Let $\mathrm{Lap}(b)$ denote the zero mean Laplace distribution with scale parameter $b>0$ (i.e., variance $2b^2$). To answer this query while satisfying $0.1$-DP, the mechanism returning the response adds $\mathrm{Lap}(10)$ noise to the true count and releases the result (see Definitions~\ref{def:dp} and~\ref{def:laplace_mech} in Section~\ref{sec:background}). However, in this case, the attacker can infer Bob's sensitive attribute with high probability by dividing the noisy answer by 10,000 and rounding to the nearest integer $j$.
\end{example}

Example~\ref{ex:bob} and similar ones underscore the necessity for privacy definitions that take into account the data-generating distribution. Consequently, several privacy frameworks have emerged to address this concern, including Pufferfish privacy~\cite{kiferPufferfishFrameworkMathematical2014}, membership privacy~\cite{liMembershipPrivacyUnifying2013}, Bayesian differential privacy~\cite{yang2015bayesian}, and coupled-worlds privacy~\cite{bassily2013coupled}, to give a few examples. Among these distribution-dependent frameworks, one that is particularly promising is based on a recent notion of information leakage called \emph{pointwise maximal leakage} (PML)~\cite{saeidian2023pointwise_it,saeidian2023pointwise_isit}. 

At its core, PML is an information measure that quantifies the amount of information leaking about a secret to a publicly available and correlated quantity. What sets PML apart is its strong operational meaning in the context of privacy. Specifically, PML is derived by assessing risks posed by adversaries in highly general threat models (see Section~\ref{sec:background} for details). Moreover, PML exhibits remarkable robustness by considering a wide range of adversaries, and flexibility in its application to various data types. Another noteworthy aspect of PML is that its guarantees and privacy parameters are easily interpretable. Informally speaking, it was shown in \cite{inferential_privacy} that enforcing privacy according to PML aligns with the fundamental principle that \say{nothing should be learnable about the secret that could not be learned from its distribution alone.} Therefore, on an abstract level, the objectives of PML parallel those of differential privacy since PML aims to reveal population-level insights about the data while concealing its intricate details. Furthermore, PML is suitable for guaranteeing privacy in complex data-processing systems through various inequalities that it satisfies, including pre-processing, post-processing, and composition inequalities, among others~\cite[Lemma 1]{saeidian2023pointwise_it}.

Interestingly, while PML is not a generalization or relaxation of differential privacy, connections have been established between the two frameworks.  More precisely, it was shown in~\cite[Thm. 4.2]{inferential_privacy} that when our goal is to protect a database containing independent entries, then differential privacy is equivalent to restricting the PML of each entry in the database across all possible outcomes of a mechanism. This result provides deep insights for protecting databases containing independent entries. However, it also prompts the question: \emph{What is the relationship between PML and differential privacy in scenarios where the entries in a database are correlated?}

\subsection{Contributions and Outline}
In this work, we establish that in scenarios where the entries in a database are correlated, the PML guarantees of mechanisms satisfying pure DP can be arbitrarily weak. Specifically, we prove that there exists a pure DP mechanism with PML levels arbitrarily close to that of a mechanism which releases individual entries from a database without any perturbation. The significance of this result lies in its quantitative nature. In particular, we rely on analytical arguments to demonstrate that DP may not be a suitable framework for protecting correlated databases, in contrast to its performance in scenarios involving independent entries. Our analysis based on PML also distinguishes our work from previous studies, which often rely on intuition, qualitative examples~\cite{kifer2012rigorous,yang2015bayesian,zhu2014correlated}, or experimental evidence~\cite{cormode2011personal,liu2016dependence} to illustrate the weakness of DP in the presence of correlations. Overall, our discussions aim to raise awareness in order to prevent creating a false sense of security that results from applying privacy mechanisms indiscriminately. 

The remainder of this paper is organized as follows: Section~\ref{sec:background} presents preliminaries, including definitions of pure differential privacy and pointwise maximal leakage, highlighting their connections and distinctions. In Section~\ref{sec:main}, we demonstrate that in scenarios involving correlated databases, mechanisms satisfying pure differential privacy can be weak when evaluated through the lens of PML. Section~\ref{sec:discussion} contains a brief discussion about our result and some concluding remarks.

\section{Preliminaries}
\label{sec:background}
\subsection{Notation}
We use uppercase letters to denote random variables, lower case letters to denote realizations, and calligraphic letters to denote sets. In particular, we reserve $X$ for representing some data containing sensitive information, e.g., a dataset containing information about individuals. For simplicity, we assume that the alphabet of $X$, denoted by $\cX$, is a finite set. With a slight abuse of notation, we use $P_X$ to describe both a distribution for $X$ and its corresponding probability mass function. Moreover, $\cP_\cX$ denotes the set of all distributions with full support on $\cX$. 

Let $Y$ be a random variable representing some information released about $X$, taking values on a set $\cY$. The random variable $Y$ is induced by a \emph{mechanism}, i.e., a conditional probability distribution $P_{Y \mid X}$. Essentially, $Y$ represents the answer to a query posed about $X$. The set $\cY$ can be finite (e.g., when $P_{Y \mid X}$ is the randomized response mechanism~\cite{warner1965randomized}) or infinite (e.g., when $P_{Y \mid X}$ is the Laplace mechanism~\cite{dworkCalibratingNoiseSensitivity}). With a slight abuse of notation, we use $P_{Y \mid X}$ to denote both the conditional distribution of $Y$ given $X$ as well as its density with respect to a suitable dominating measure, e.g., the counting measure or the Lebesgue measure.  

Let $P_{XY}$ denote the joint distribution of $X$ and $Y$. We write $P_{XY} = P_{Y \mid X} \times P_X$ to imply that $P_{XY}(x,y) = P_{Y \mid X =x}(y) P_X(x)$ for all $(x,y) \in \mathcal X \times \mathcal Y$. Furthermore, we write $P_{Y} = P_{Y \mid X} \circ P_X$ to represent marginalization over $X$, i.e., to imply that ${P}_{Y}(y) = \sum_{x \in \mathcal X}  P_{Y \mid X=x}(y)  P_X(x)$ for all $y \in \mathcal Y$. 

We say that the Markov chain $U - X - Y$ holds if random variables $U$ and $Y$ are conditionally independent given $X$, that is, if $P_{UY \mid X} = P_{U \mid X} \times P_{Y \mid X}$. The Markov chain $U - X - Y$ implies that  $Y$ depends on $U$ only through $X$ and vice versa. We may think of a $U$ satisfying the Markov chain $U-X-Y$ as either a feature of $X$ or a (randomized) function of $X$. 

Finally, $[n] \coloneqq \{1, \ldots, n\}$ describes the set of all positive integers smaller than or equal to $n$, and $\log(\cdot)$ denotes the natural logarithm.

\subsection{Differential Privacy}
Often called the gold standard of privacy, differential privacy (DP)~\cite{dworkCalibratingNoiseSensitivity,dwork2014algorithmic} stands as the most widely adopted privacy framework, both in theoretical developments as well as real-world deployments. Conceptually, DP considers scenarios where individuals' data is aggregated into a database, with the aim of responding to queries posed to the database without compromising the privacy of any individual contributor. To achieve this goal, DP guarantees that two databases that differ in only one entry (referred to as \say{neighboring} databases) produce query responses that are hard to distinguish. 

Over nearly two decades of extensive research has led to the development of various DP variants and adaptations, such as approximate DP~\cite{dworkOurDataOurselves2006a}, concentrated DP~\cite{bun2016concentrated,dwork2016concentrated}, Rényi DP~\cite{mironov2017renyi}, and Gaussian DP~\cite{dong2022gaussian}, among others. However, in this paper, our focus is on the original, and arguably the strongest, form of DP, known as \say{pure} DP~\cite{dworkCalibratingNoiseSensitivity}.

Let $X = (D_1, \ldots, D_n)$ be a database containing $n$ entries. Given $i \in [n]$, $D_i$ represents the $i$-th entry, which takes values on a finite set $\mathcal D$ and $D_{-i} = (D_1, \ldots, D_{i-1}, D_{i+1}, \ldots, D_n)$ represents the database with its $i$-th entry removed. Let $P_X = P_{D_1, \ldots, D_n}$ be the distribution according to which databases are drawn from $\cX = \mathcal D^n$. To obtain the distribution of the $i$-th entry, we marginalize over the remaining $n-1$ entries, that is, for each $d_i \in \mathcal D$ and $i \in [n]$ we have
\begin{align}
\begin{split}
\label{eq:marginal}
    P_{D_i}(d_i) &= \sum_{d_{-i} \in \mathcal D^{n-1}} P_X(d_i, d_{-i})\\
    &= \sum_{d_{-i} \in \mathcal D^{n-1}} P_{D_i \mid D_{-i} = d_{-i}}(d_i) \; P_{D_{-i}} (d_{-i}),
\end{split}
\end{align}
where $d_{-i} \coloneqq (d_1, \ldots, d_{i-1}, d_{i+1}, \ldots, d_n) \in \mathcal D^{n-1}$ is a tuple describing the database with its $i$-th entry removed. Note that this setup is very general since we make no independence assumptions and the entries can be arbitrarily correlated. 

Suppose an analyst poses a query to the database, with the answer returned by the mechanism $P_{Y \mid X}$.

\begin{definition}[Differential privacy]
\label{def:dp}
Given $\varepsilon \geq 0$, the mechanism $P_{Y \mid X}$ satisfies $\varepsilon$-differential privacy if 
\begin{equation*}
    \max_{\substack{d_i, d_i' \in \mathcal D: \\ i \in [n]}} \; \max_{d_{-i} \in \mathcal D^{n-1}} \log \frac{P_{Y \mid D_i=d_i, D_{-i} = d_{-i}}(\cE)}{P_{Y \mid D_i=d_i', D_{-i} = d_{-i}}(\cE)} \leq \varepsilon,
\end{equation*}
for all measurable events $\cE \subseteq \cY$. 
\end{definition}

Importantly, Definition~\ref{def:dp} is agnostic to the distribution of the database, $P_X$. Therefore, differential privacy is a property of the mechanism $P_{Y \mid X}$ alone. This observation has led to the widespread belief that DP offers strong privacy guarantees for individuals in a database irrespective of the distribution $P_X$. However, in Section 3 we will present a quantitative example to demonstrate that a differentially private mechanism can in fact leak a large amount of information about individual entries in scenarios where the entries are correlated. 

Next, let us recall one of the most commonly employed differentially private mechanisms, namely the \emph{Laplace mechanism}~\cite{dworkCalibratingNoiseSensitivity}. The Laplace mechanism is often used to answer numerical queries with bounded $\ell_1$-\emph{sensitivity}. Let $\mathrm{Lap}(\mu, b)$ denote the Laplace distribution with mean $\mu \in \bR$ and scale parameter $b>0$ (i.e., variance $2b^2$). 

\begin{definition}[Laplace mechanism]
\label{def:laplace_mech}
Let $f : \cX \to \bR$ be a query with $\ell_1$-sensitivity 
\begin{equation*}
    \Delta_1(f) \coloneqq \sup_{x_1, x_2 \in \cX : x_1 \sim x_2} \abs{f(x_1) - f(x_2)}.
\end{equation*}
Given $\varepsilon > 0$ the Laplace mechanism returns a query response according to the distribution $Y \mid X=x \sim \mathrm{Lap}\left(f(x), \frac{\Delta_1(f)}{\varepsilon} \right)$, where $x \in \cX$. 
\end{definition}

It has been shown in~\cite{dworkCalibratingNoiseSensitivity} that the Laplace mechanism satisfies $\varepsilon$-DP.

The $\ell_1$-sensitivity of a query $f$ describes the largest change in $f(x)$ upon altering the value of one entry in database $x$. The Laplace mechanism then computes $f$ and perturbs it with Laplace noise scaled according to $\Delta_1(f)$ and $\varepsilon$. Examples of queries that can be answered via the Laplace mechanism include \emph{counting queries}, that is, queries of the form \say{How many entries in the database satisfy property $A$?,} and \emph{histogram queries}~\cite{dwork2014algorithmic}.

\subsection{Pointwise Maximal Leakage}


Pointwise maximal leakage (PML)~\cite{saeidian2023pointwise_it,saeidian2023pointwise_isit} is a recently introduced privacy measure that enjoys a strong operational meaning and robustness. It measures the amount of information leaking about a secret $X$ to the outcomes of a mechanism $P_{Y \mid X}$. PML is obtained by evaluating the risk posed by adversaries in two highly versatile threat models: the \emph{randomized functions} model~\cite{issaOperationalApproachInformation2020} and the \emph{gain function} model of leakage~\cite{alvim2012measuring}.

Below, we define PML using both models.
\begin{definition}[Randomized function view of PML~{\cite[Def. 1]{saeidian2023pointwise_it}}]
\label{def:randomized_function_view}
Suppose $X$ is a random variable on the finite set $\cX$, and $Y$ is a random variable on a set $\cY$ induced by a mechanism $P_{Y \mid X}$. 
According to the randomized function view, the pointwise maximal leakage from $X$ to $y \in \cY$ is defined as
\begin{equation}
\label{eq:pml_u_sup}
    \ell(X\to y) \coloneqq \log \sup_{U:U-X-Y} \frac{\sup_{P_{\hat U \mid Y}} \mathbb P \left(U=\hat U \mid Y=y \right)}{\max_{u \in \cU} P_U(u)},
\end{equation}
where $U$ and $\hat U$ are random variables on a finite set $\cU$, $P_{U} = P_{U \mid X} \circ P_X$, and the Markov chain $U-X-Y-\hat U$ holds.
\end{definition}

Definition~\ref{def:randomized_function_view} can be understood as follows. Let $U$ be a randomized function of $X$. For instance, when $X$ is a database, $U$ can represent a single entry or a subset of the entries in $X$. To quantify the amount of information leakage associated with a single released outcome of the mechanism, denoted by $y$, in \eqref{eq:pml_u_sup} we compare the probability of correctly guessing the value of $U$ after observing $y$ in the numerator of~\eqref{eq:pml_u_sup} with the \emph{a priori} probability of correctly guessing the value of $U$ in the denominator. The probability of correctly guessing $U$ after observing $y$ is assessed by assuming that the adversary uses the best guessing kernel $P_{\hat U \mid Y}$, represented by the supremum over $P_{\hat U \mid Y}$ in the numerator of~\eqref{eq:pml_u_sup}. Similarly, the prior probability of correctly guessing $U$ is $\max_{u \in \mathcal{U}} P_U(u)$. Crucially, Definition~\ref{def:randomized_function_view} results in a highly robust measure of privacy since the posterior-to-prior ratio is maximized over all possible randomized functions of $X$, represented by the supremum over all $U$'s satisfying the Markov chain $U-X-Y$. This makes PML particularly useful when we do not know what feature of $X$ an adversary is interested to guess, or different adversaries may be interested in different features of $X$. 

Next, we define PML using the gain function model. 
\begin{definition}[Gain function view of PML~{\cite[Cor. 1]{saeidian2023pointwise_it}}]
\label{def:gain_function_view}
Suppose $X$ is a random variable on the finite set $\cX$, and $Y$ is a random variable on a set $\cY$ induced by a mechanism $P_{Y \mid X}$. 
According to the gain function view, the pointwise maximal leakage from $X$ to $y \in \cY$ is defined as
\begin{equation}
\label{eq:g-leakage}
    \ell(X\to y) \coloneqq \log \; \sup_g \; \frac{\sup_{P_{W \mid Y}} \mathbb{E} \left[g(X,W) \mid Y=y \right]}{\max_{w \in \cW} \bE\left[g(X, w)\right]},
\end{equation}
where the supremum is over all non-negative gain functions $g: \cX \times \cW \to \bR_+$ with a finite range, and $\cW$ is a finite set.  
\end{definition}

Definition~\ref{def:gain_function_view} can be understood as follows. Consider an adversary whose objective is to construct a guess of $X$, denoted by $W$, in order to maximize the expected value of a non-negative gain function $g$. The gain function $g$ captures the adversary's objective and can be tailored to model a wide array of privacy attacks. For example, when $X$ is a database, $g$ can model \emph{membership inference} attacks or \emph{reconstruction} attacks~\cite{dwork2017exposed} (see \cite{saeidian2023pointwise_it} for concrete examples of gain functions). To quantify the amount of information leakage associated with a single released outcome $y$, we compare the expected value of $g$ after observing $y$ in the numerator of~\eqref{eq:g-leakage} with the prior expected value of $g$ in the denominator. The posterior expected gain is assessed using the best kernel $P_{W \mid Y}$, represented by the supremum over $P_{W \mid Y}$ in the numerator of~\eqref{eq:g-leakage}. Similarly, the prior expected gain is $\max_{w \in \cW} \bE\left[g(X, w)\right]$. Then, to obtain a privacy measure robust to different types of attacks, the posterior-to-prior ratio of the expected gain is maximized over all possible non-negative $g$'s with a finite range. 

While Definitions \ref{def:randomized_function_view} and \ref{def:gain_function_view} offer different approaches to defining PML, we showed in \cite[Thm. 2]{saeidian2023pointwise_it} that, in fact, they are mathematically equivalent. Moreover, both definitions can be simplified to the following concise expression.
\begin{theorem}[{\cite[Thm. 1]{saeidian2023pointwise_it}}]
\label{thm:pml}
Let $P_{XY}$ be a distribution on the set $\mathcal X \times \mathcal Y$ with the marginal distribution $P_X \in \cP_\cX$ for $X$. The pointwise maximal leakage from $X$ to $y \in \mathcal Y$ is\footnote{We use the convention that $P_{X \mid Y=y} = P_X$ if $P_Y(y) =0$. That is, conditioning on outcomes with density zero equals no conditioning.}
\begin{equation}
\label{eq:pml_simple}
    \ell(X \to y) = D_\infty(P_{X \mid Y=y} \Vert P_X),
\end{equation}
where $P_{X \mid Y=y}$ denotes the posterior distribution of $X$ given $y \in \mathcal Y$, and
\begin{align*}
    D_\infty(P_{X \mid Y=y} \Vert P_X) &= \log \; \max_{x \in \cX} \; \frac{P_{X \mid Y=y}(x)}{P_X(x)}\\
    &= \log \; \max_{x \in \cX} \; \frac{P_{Y \mid X=x}(y)}{P_Y(y)},
\end{align*}
denotes the \Ren divergence of order infinity~\cite{renyi1961measures} of $P_{X \mid Y=y}$ from $P_X$.  
\end{theorem}

In addition to its strong operational meaning and robustness, PML satisfies several useful properties that render it suitable for deployment in complex data-processing systems. Notably, PML satisfies a pre-processing inequality, a post-processing inequality, and increases (at most) linearly under composition~\cite[Lemma 1]{saeidian2023pointwise_it}. Furthermore, as evident from~\eqref{eq:pml_simple}, PML is non-negative and satisfies the bound
\begin{equation}
\label{eq:pml_bounds}
    \ell(X \to y) \leq \log \; \frac{1}{\min\limits_{x \in \cX} P_X(x)},
\end{equation}
for all mechanisms and all $y \in \cY$. The right hand side of the above inequality essentially describes the maximum amount of information that can leak about $X$ through any mechanism. In other words, it represents the largest PML across all outcomes of a mechanism that releases $X$ without perturbing it.

\subsection{Differential Privacy as a PML Constraint}
In general, PML and DP offer fundamentally distinct approaches to privacy and differ in several key aspects. PML quantifies the amount of information leaked to an outcome of a privacy mechanism and the secret $X$ may encompass various types of sensitive data. For example, $X$ can be a password, an individual's medical records, or an entire database. In contrast, DP was specifically formulated to protect private databases. More importantly, PML depends on both the mechanism $P_{Y \mid X}$ and the data-generating distribution $P_X$. Consequently, a mechanism that leaks little information leakage under one distribution may leak a lot of information under another distribution. Conversely, DP depends only on the mechanism. 

Despite their differences, in~\cite[Thm. 4.2]{inferential_privacy}, we established that when $X$ is a database containing independent entries, DP is equivalent to restricting the amount of information leaked about each entry across all outcomes of a mechanism. Let $\cX = \cD^n$ denote the set of all possible databases, and $\mathcal Q_\mathcal X$ denote the set of product distributions in $\mathcal P_{\mathcal X}$ defined as $\mathcal Q_\mathcal X \coloneqq \{P_{X} \in \mathcal P_{\mathcal X} \colon P_{X} = \prod_{i=1}^n P_{D_i} \}$.

\begin{theorem}[{\cite[Thm. 4.2]{inferential_privacy}}]
\label{thm:dp_pml}
Given $\varepsilon \geq 0$, the mechanism $P_{Y \mid X}$ satisfies $\varepsilon$-differential privacy if and only if 
\begin{equation*}
    \sup_{P_X \in \mathcal Q_\mathcal X} \; \sup_{y \in \mathcal Y} \; \max_{i \in [n]} \; \ell(D_i \to y) \leq \varepsilon. 
\end{equation*}
\end{theorem}

Theorem~\ref{thm:dp_pml} demonstrates that when the entries in a database are independent, DP is adequate for ensuring privacy in the sense of PML. However, it also raises the possibility that when the entries are correlated, DP might fall short in terms of PML. We delve into this topic in the next section.

\section{Privacy for Correlated Databases: PML vs. DP}
\label{sec:main}
As discussed in Section~\ref{sec:intro}, the objective of our work is to understand the relation between PML and DP when $X$ is a database with correlated entries. It turns out that pure DP mechanisms can have poor privacy performance when assessed through the lens of PML. Before we formally state our result, we need to define the PML of a mechanism that releases an entry from the database without perturbation. 

Let $X = (D_1, \ldots, D_n)$ be a database. Given a distribution $P_X \in \cP_\cX$ over $X$ and $i \in [n]$, let $P_{D_i}$ denote the distribution of the $i$-th entry in a database, obtained by marginalizing $P_X$ over $P_{D_{-i}}$, described by~\eqref{eq:marginal}. We use
\begin{equation*}
    \varepsilon_\mathrm{max} (D_i) \coloneqq \log \frac{1}{\min\limits_{d \in \cD} \, P_{D_i}(d)},
\end{equation*}
 to denote the largest amount of information that can leak about $D_i$ through any mechanism. By \eqref{eq:pml_bounds}, $\varepsilon_\mathrm{max} (D_i)$ is equal to the PML of a mechanism that releases $D_i$ with no randomization.  

\begin{theorem}
\label{thm:main}
For each $\delta >0$ and $\varepsilon >0$ there exists a database $X = (D_1, \ldots, D_n)$, $i \in [n]$, a mechanism $P_{Y \mid X}$ satisfying $\varepsilon$-DP, and $y \in \bR$ such that
\begin{equation*}
    \ell(D_i \to y) > \varepsilon_\mathrm{max}(D_i) - \delta.
\end{equation*}
\end{theorem}

\begin{proof}
To prove the statement, we construct a binary database where the entries are strongly correlated: If one entry is zero (resp. one), then the other entries are also likely to be zero (resp. one). We then demonstrate that if we use the Laplace mechanism to answer the counting query on this database, then the resulting PML is significantly high.

Let $X = (D_1, \ldots, D_{n +1})$ be a database containing $n+1$ binary entries.\footnote{We consider a database of size $n+1$ instead of $n$ for notational convenience.} Suppose $P_{D_1}(0) = 1 - P_{D_1}(1) = \alpha$ with $0 < \alpha < 0.5$. Let $D_-=(D_2, \ldots, D_{n+1})$. Fix a constant $0 < \eta < 1$ and suppose the distribution of $D_-$ depends on $D_1$ as follows:  
\begin{gather*}
    P_{D_- \mid D_1=1}(d_-) = \begin{cases}
        \eta, &\quad \mathrm{if} \; d_- = 1^{n},\\
        \frac{1 - \eta}{2^{n} -1}, &\quad \mathrm{otherwise}.
    \end{cases}\\
    P_{D_- \mid D_1=0}(d_-) = \begin{cases}
        \eta, &\quad \mathrm{if} \; d_- = 0^{n},\\
        \frac{1 - \eta}{2^{n} -1}, &\quad \mathrm{otherwise}.
    \end{cases}
\end{gather*}
Suppose our goal is to release the empirical frequency of the ones in the database using the Laplace mechanism, i.e., $Y \mid X=x \sim \mathrm{Lap}(\frac{\norm{x}_1}{n+1}, b)$, where $\norm{x}_1$ denotes the $\ell_1$-norm of $x \in \{0,1\}^{n+1}$. Note that the empirical frequency has global sensitivity $\frac{1}{n+1}$, thus the Laplace mechanism with scale parameter $b = \frac{1}{\varepsilon (n+1)}$ satisfies $\varepsilon$-DP. However, here we show that the Laplace mechanism is insufficient for protecting $D_1$. To demonstrate this, we calculate the PML $\ell(D_1 \to y)$ with $y \leq 0$, which depends on the distributions $P_{Y \mid D_1=1}, P_{Y \mid D_1=0}$ and $P_{Y}$. 

First, we calculate $P_{Y \mid D_1=1}(y)$ assuming $y \leq 0$:
\begin{align*}
    &P_{Y \mid D_1=1}(y) = \sum_{d_- \in \{0,1\}^{n}} P_{Y \mid D_1=1, D_-=d_-}(y) \cdot P_{D_- \mid D_1=1}(d_-)\\
    &= \eta \; P_{Y \mid D_1=1, D_-= 1^n}(y) + \frac{1-\eta}{2^{n} -1} \sum_{d_- \in \{0,1\}^{n} \setminus 1^{n}} P_{Y \mid D_1=1, D_-=d_-}(y)\\
    &= \frac{\eta}{2b} \exp \left(- \frac{\abs{y - 1}}{b} \right) + \frac{1-\eta}{2b (2^{n} -1)} \sum_{d_- \in \{0,1\}^{n} \setminus 1^{n}} \exp \left(- \frac{\abs{y - \frac{1}{n+1} - \frac{\norm{d_-}_1}{n+1}}}{b} \right)\\
    &= \frac{\eta}{2b} \exp \left(- \frac{\abs{y - 1}}{b} \right) + \frac{1-\eta}{2b (2^{n} -1)} \cdot \sum_{i=0}^{n-1} \binom{n}{i} \exp \left(- \frac{\abs{y - \frac{1}{n+1} - \frac{i}{n+1}}}{b} \right)\\
    &= \frac{\eta}{2b} \exp \left(- \frac{\abs{y - 1}}{b} \right) + \\
    &\hspace{3em} \frac{1-\eta}{2b (2^{n} -1)} \bigg[ \sum_{i=0}^{n} \binom{n}{i} \exp \left(- \frac{\abs{y - \frac{1}{n+1} - \frac{i}{n+1}}}{b} \right) - \exp \left(- \frac{\abs{y -1}}{b} \right) \bigg]\\
    &= \frac{2^n \eta - 1}{2b (2^n - 1)} \exp \left(\frac{y - 1}{b} \right) + \\
    &\hspace{.3\textwidth} \frac{1-\eta}{2b (2^{n} -1)} \exp \left(\frac{y - \frac{1}{n+1}}{b}\right) \cdot \left(1 + \exp (- \frac{1}{b(n+1)}) \right)^n\\
    &= \frac{1}{2b (2^{n} -1)}\exp(\frac{y}{b}) \bigg[\Big(2^n \eta - 1\Big) \exp \left(-\frac{1}{b} \right) + \\
    &\hspace{.3\textwidth} (1-\eta) \, \exp \left(-\frac{1}{b(n+1)}\right) \left(1 + \exp (- \frac{1}{b(n+1)}) \right)^n \bigg].
\end{align*}

Similarly, we can calculate $P_{Y \mid D_1=0}(y)$ assuming $y \leq 0$:  
\begin{align*}
    &P_{Y \mid D_1=0}(y) = \sum_{d_- \in \{0,1\}^{n}} P_{Y \mid D_1=0, D_-=d_-}(y) \cdot P_{D_- \mid D_1=0}(d_-) \\
    &= \eta \; P_{Y \mid D_1=0, D_-= 0^n}(y) + \frac{1-\eta}{2^{n} -1} \; \sum_{d_- \in \{0,1\}^{n} \setminus 0^{n}} P_{Y \mid D_1=0, D_-=d_-}(y)\\
    &= \frac{\eta}{2b} \exp \left(\frac{y}{b} \right) + \frac{1-\eta}{2b (2^{n} -1)} \; \sum_{d_- \in \{0,1\}^{n} \setminus 0^{n}} \exp \left(\frac{y - \frac{\norm{d_-}_1}{n+1}}{b} \right)\\
    &= \frac{1}{2b} \exp \left(\frac{y}{b} \right) \left[\eta + \frac{1-\eta}{2^{n} -1}  \sum_{i=1}^{n} \binom{n}{i} \exp \left(- \frac{i}{b(n+1)} \right) \right] \\  
    &= \frac{1}{2b} \exp \left(\frac{y}{b} \right) \left[\eta - \frac{1-\eta}{2^{n} -1} + \frac{1-\eta}{2^{n} -1}  \sum_{i=0}^{n} \binom{n}{i} \exp \left(- \frac{i}{b(n+1)} \right) \right] \\  
    &= \frac{1}{2b (2^{n} -1)} \exp \left(\frac{y}{b} \right) \left[2^n  \eta - 1 + (1-\eta) \left(1 + \exp (- \frac{1}{b(n +1)}) \right)^n\right].  \\  
\end{align*}
Since $\exp(-x) \leq 1$ for $x \geq 0$, then $P_{Y \mid D_1=1}(y) \leq P_{Y \mid D_1=0}(y)$ when $y \leq 0$. Next, we calculate $P_Y(y)$ for $y \leq 0$: 
\begin{align*}
    &P_Y(y) = (1 - \alpha) P_{Y \mid D_1 = 1}(y) + \alpha P_{Y \mid D_1=0}(y) \\
    &= \frac{1}{2b (2^{n} -1)} \exp \left(\frac{y}{b} \right) \Bigg[\Big(2^n \eta - 1\Big) (1 - \alpha) \exp \left(-\frac{1}{b} \right) + \\
    &\hspace{7em} (1-\eta) (1 - \alpha)\, \exp \left(-\frac{1}{b(n+1)}\right) \left(1 + \exp (- \frac{1}{b(n+1)}) \right)^n + \\
    &\hspace{7em} \big(2^n \eta - 1\big) \alpha +  (1-\eta) \alpha \left(1 + \exp (- \frac{1}{b(n +1)}) \right)^n \Bigg]\\
    &\leq \frac{1}{2b (2^{n} -1)} \exp \left(\frac{y}{b} \right) \Bigg[\Big(2^n \eta - 1\Big) \Big( (1 - \alpha)\, \exp \left(-\frac{1}{b}\right) + \alpha \Big) + \\
    &\hspace{1em} (1-\eta) (1 - \alpha)\, \left(1 + \exp (- \frac{1}{b(n+1)}) \right)^n + (1-\eta) \alpha \left(1 + \exp (- \frac{1}{b(n +1)}) \right)^n \Bigg]\\
    &= \frac{1}{2b (2^{n} -1)} \exp \left(\frac{y}{b} \right) \Bigg[\Big(2^n \eta - 1\Big) \Big( (1 - \alpha)\, \exp \left(-\frac{1}{b}\right) + \alpha \Big) +\\
    &\hspace{18em}(1-\eta) \left(1 + \exp (- \frac{1}{b(n+1)}) \right)^n  \Bigg]\\
    &\leq \frac{1}{2b (2^{n} -1)} \exp \left(\frac{y}{b} \right) \Bigg[2^n  \eta \Big( (1 - \alpha)\, \exp \left(-\frac{1}{b}\right) + \alpha \Big) +\\
    &\hspace{18em} (1-\eta) \left(1 + \exp (- \frac{1}{b(n+1)}) \right)^n  \Bigg]. 
\end{align*}
Using $b(n+1) = \frac{1}{\varepsilon}$ to achieve $\varepsilon$-DP, we obtain the following lower bound on $\ell(D_1 \to y)$ with $y \leq 0$:
\begin{align*}
    \ell(D_1 \to y) &= \log \frac{\max\limits_{d_1 \in \{0,1\}} P_{Y \mid D_1 = d_1}(y)}{P_Y(y)} = \log \frac{P_{Y \mid D_1 = 0}(y)}{P_Y(y)}\\[0.5em]
    &\geq \log \frac{2^n \eta + \left(1 + e^{-\varepsilon} \right)^n (1-\eta) - 1}{2^n  \eta \Big( (1 - \alpha)\, \exp (-\varepsilon  n-\varepsilon) + \alpha \Big) + \left(1 + e^{-\varepsilon} \right)^n (1-\eta)}\\[0.5em]
    &= \log \frac{2^n  \eta +  \left(1 + e^{-\varepsilon} \right)^n \, (1-\eta) - 1}{2^n  \eta \, \alpha  + \left(\frac{2}{e^\varepsilon}\right)^n \eta (e^{-\varepsilon}) (1 - \alpha) +  \left(1 + e^{-\varepsilon} \right)^n \, (1-\eta) }.
 \end{align*}    

Note that $1 + e^{-\varepsilon} < 2$ and $\frac{2}{e^\varepsilon} < 2$ for all $\varepsilon > 0$. Therefore, when $n$ is large the dominating term in the numerator is $2^n \eta$ and the dominating term in the denominator is $2^n \eta \, \alpha$. Hence, as $n \to \infty$, the lower bound on $\ell(D_1 \to y)$ approaches $\varepsilon_\mathrm{max} (D_1) = \log \frac{1}{\alpha}$. This proves that for each $\delta >0$, there exists an integer $n$, a database $X$ of size $n$, and a mechanism satisfying $\varepsilon$-DP such that 
\begin{equation*}
    \ell(D_1 \to y) > \varepsilon_\mathrm{max}(D_1) - \delta.
\end{equation*}
\qed
\end{proof}

The proof of Theorem~\ref{thm:main} relies on a database exhibiting what may be considered as pathologically strong correlations: If the first entry is zero (resp. one) then all other entries are likely to be zero (resp. one) with a constant probability that does not diminish with growing database size $n$. However, it is important to note that the theorem holds true even in more realistic scenarios characterized by weaker correlations. Specifically, the asymptotic lower bound of $\varepsilon_\mathrm{max}(D_1)$ for PML remains applicable even if $\eta$ diminishes at a polynomial rate, i.e., if $\eta = \Theta(\frac{1}{n^r})$ for some constant $r \geq 1$.

\section{Conclusions}
\label{sec:discussion}
While previous research has hinted at the inadequacy of differential privacy in protecting databases with correlated entries, our analysis using pointwise maximal leakage offers a quantitative evaluation of the potential weaknesses in DP's guarantees. Our investigation illuminates the vulnerabilities arising from DP's distribution-agnostic definition, even in its strongest form, i.e., pure DP, which could foster a false sense of privacy if applied indiscriminately. In contrast, PML's ability to adjust to diverse scenarios due to its distribution-dependent nature provides a more nuanced approach to privacy preservation. In this way, our work highlights the necessity for further research aimed at developing a wide array of PML mechanisms tailored to specific contexts.

\begin{credits}
\subsubsection{\ackname} This work has been supported by the Swedish Research Council (VR) under the grant 2023-04787 and Digital Futures center within the collaborative project DataLEASH.

\subsubsection{\discintname}
The authors have no competing interests to declare that are
relevant to the content of this article.
\end{credits}
%
%
%
\bibliographystyle{splncs04}
\bibliography{main}

\end{document}